\documentclass[11pt]{article}
\usepackage[margin=1in]{geometry}
\usepackage{cite,microtype,graphicx}
\usepackage{hyperref,aliascnt}
\usepackage{amsmath,amssymb,amsfonts,amsthm}

\newtheorem{theorem}{Theorem}

\newaliascnt{lemma}{theorem} \newtheorem{lemma}[lemma]{Lemma}
\aliascntresetthe{lemma} 

\newaliascnt{corollary}{theorem}

\aliascntresetthe{corollary}

\newaliascnt{definition}{theorem}

\aliascntresetthe{definition}

\newaliascnt{observation}{theorem}
\newtheorem{observation}[observation]{Observation}
\aliascntresetthe{observation}

\title{On Polyhedral Realization with Isosceles Triangles}

\author{David Eppstein\\
Computer Science Department, University of California, Irvine}

\date{ }

\begin{document}

\maketitle
\begin{abstract}
Answering a question posed by Joseph Malkevitch, we prove that there exists a polyhedral graph, with triangular faces, such that every realization of it as the graph of a convex polyhedron includes at least one face that is a scalene triangle. Our construction is based on Kleetopes, and shows that there exists an integer $i$ such that all convex $i$-iterated Kleetopes have a scalene face. However, we also show that all Kleetopes of triangulated polyhedral graphs have non-convex non-self-crossing realizations in which all faces are isosceles. We answer another question of Malkevitch by observing that a spherical tiling of Dawson (2005) leads to a fourth infinite family of convex polyhedra in which all faces are congruent isosceles triangles, adding one to the three families previously known to Malkevitch. We prove that the graphs of convex polyhedra with congruent isosceles faces have bounded diameter and have dominating sets of bounded size.
\end{abstract}

\section{Introduction}

By Steinitz's theorem, the graphs of three-dimensional convex polyhedra are exactly the 3-vertex-connected planar graphs~\cite{Ste-EMW-22}. The faces of the polyhedron are uniquely determined as the peripheral cycles of the graph: simple cycles such that every two edges that are not in the cycle can be connected by a path whose interior vertices are disjoint from the cycle~\cite{Tut-PLMS-63}. We call such a graph a \emph{polyhedral graph}, and if in addition all faces are triangles we call it a \emph{triangulation}. If $G$ is the graph of a polyhedron $P$, we call $P$ a \emph{realization} of $G$; here, $P$ is required only to be non-self-intersecting, not necessarily convex, but if it is convex we call $P$ a \emph{convex realization}. We may ask: what constraints are possible on the shapes of the faces of $P$~\cite{BarGru-PJM-70}?
 
\begin{figure}[t]
\centering\includegraphics[width=0.35\textwidth]{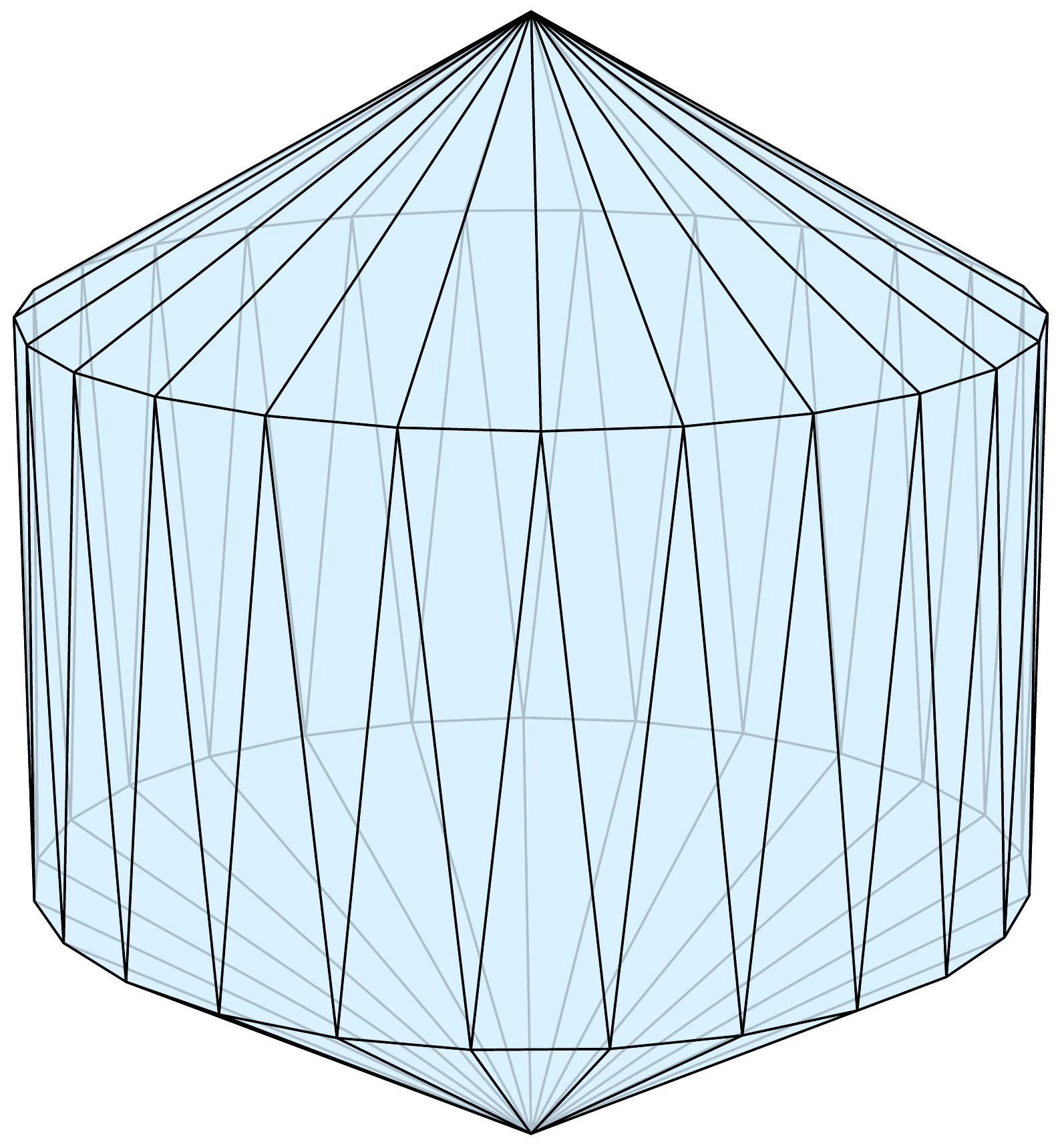}\qquad
\includegraphics[width=0.4\textwidth]{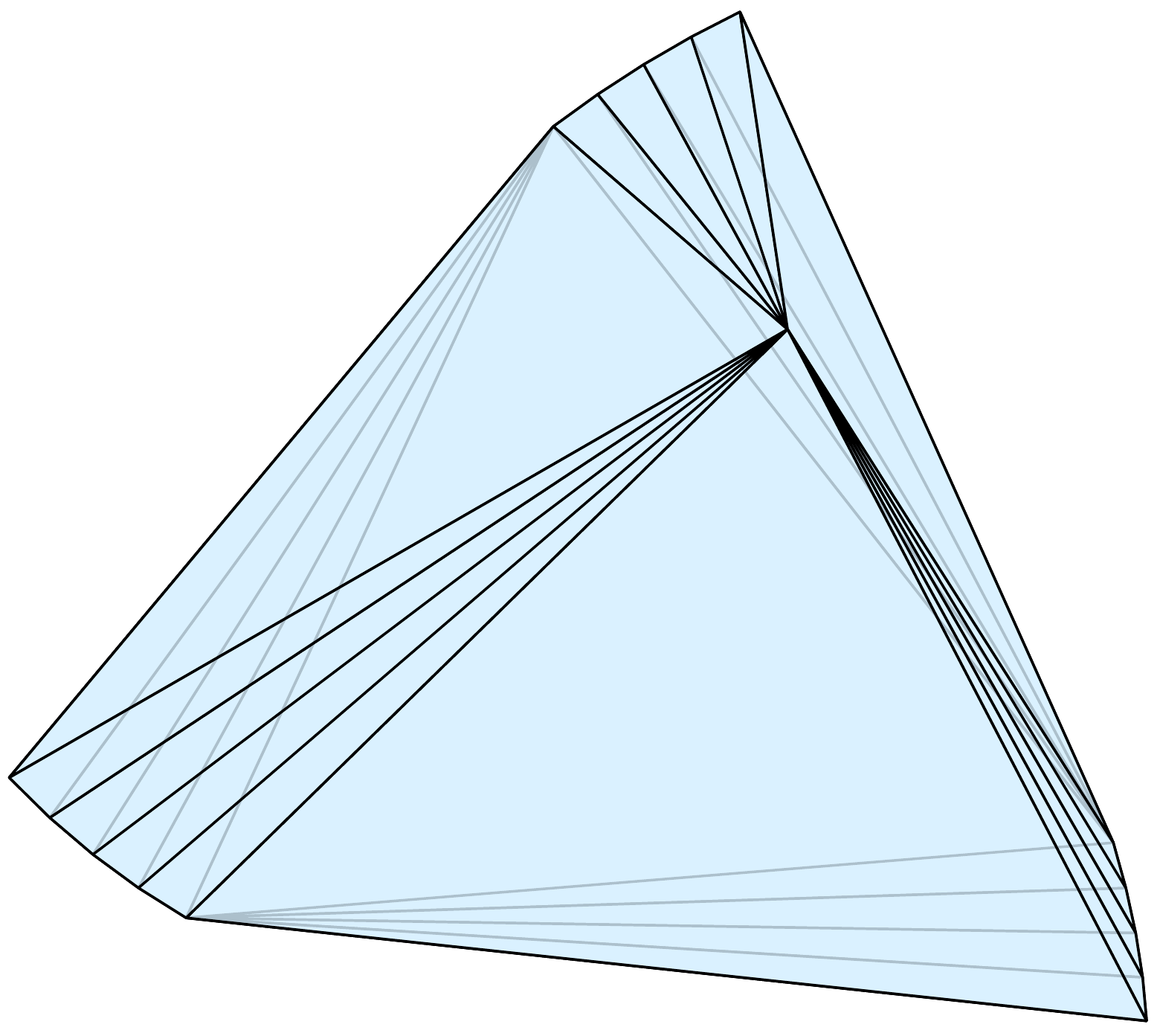}
\caption{Left: A polyhedron of Goldberg~\cite{Gol-AMM-36}, with 100 congruent isosceles triangle faces.
Right: A triangulation of Gr\"unbaum that cannot be realized with all faces congruent~\cite{Gru-Geomb-01}, realized as a convex polyhedron with edges of two lengths, and with isosceles and equilateral triangle faces.}
\label{fig:grunbaum}
\end{figure}

We call a realization or polyhedron \emph{isosceles} if all faces are isosceles triangles, and \emph{monohedral} if all faces are congruent.\footnote{This terminology should be distinguished from \emph{isohedral} polyhedra, in which every two faces are symmetric to each other.}  Monohedral isosceles convex polyhedra were investigated by Goldberg in 1936, who observed the existence of an infinite family of these polyhedra (beyond the obvious example of the bipyramids), obtained by gluing pyramids to the non-triangular faces of antiprisms~\cite{Gol-AMM-36} (\autoref{fig:grunbaum}, left). 
In 2001, Joseph Malkevitch described a triangulation that has no monohedral isosceles convex realization~\cite{Mal-Geomb-01}; in response, Branko Gr\"unbaum described a triangulation that has no monohedral realization, regardless of convexity of the realization and regardless of face shape~\cite{Gru-Geomb-01}. However, Gr\"unbaum's triangulation can be realized with all faces isosceles or equilatera ( \autoref{fig:grunbaum}, right). More recently, Malkevitch has asked again which triangulations have isosceles convex realizations, not required to be congruent, and whether all triangulations have such a realization~\cite{Mal-COMAP-19}.  Isosceles polyhedra have also been studied in connection with unfolding polyhedra into nets: there exist non-convex isosceles polyhedra with all faces acute that have no such unfolding, while it remains open whether all convex polyhedra have an unfolding~\cite{BerDemEpp-CGTA-03,DemDemEpp-CCCG-20}.

In this work, we provide the following results:
\begin{itemize}
\item We prove that some triangulations have no convex isosceles realization. We use the idea of a \emph{Kleetope}, a construction named after Victor Klee in which a polyhedron or polyhedral graph is modified by gluing a pyramid onto each face~\cite{Gru-IJM-63}. We show that repeating the Kleetope construction a bounded number of times, starting from any polyhedral graph, will produce a triangulation that has no convex isosceles realization.
\item We prove that the triangulations produced by this construction, and more generally all the Kleetopes of triangulations, have non-convex isosceles realizations.
\item We answer an open question posed by Malkevitch~\cite{Mal-COMAP-19} by finding a fourth infinite family of monohedral isosceles convex polyhedra, different from three previously known infinite families.
\item We investigate the structure of monohedral isosceles convex polyhedra with many faces, proving that they must have very sharp apex angles, bounded graph diameter, and bounded dominating set size.
\end{itemize}

\section{Preliminaries}

\subsection{Kleetopes and shortness}
\begin{figure}[t]
\centering\includegraphics[width=0.4\textwidth]{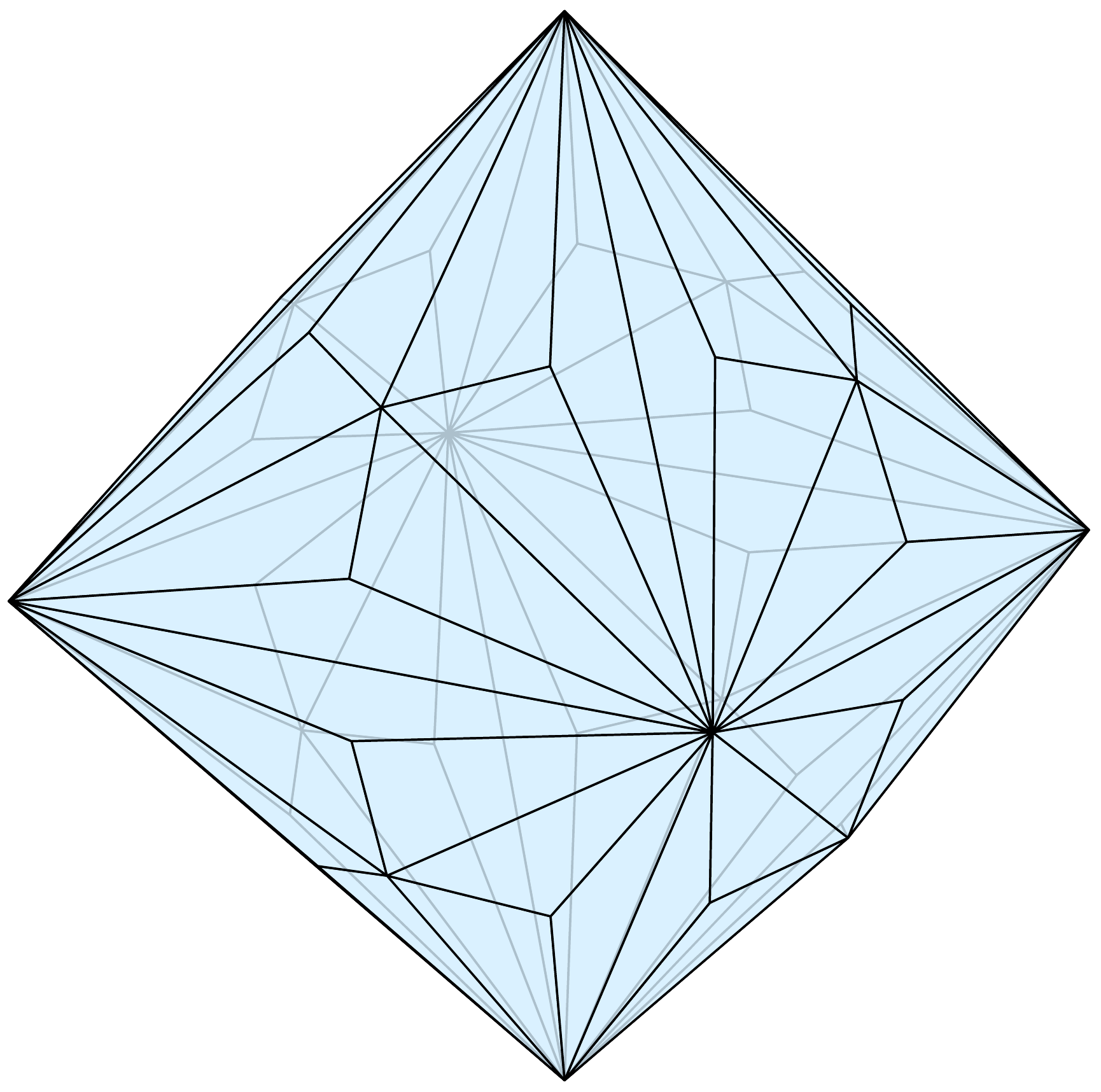}\qquad\includegraphics[width=0.4\textwidth]{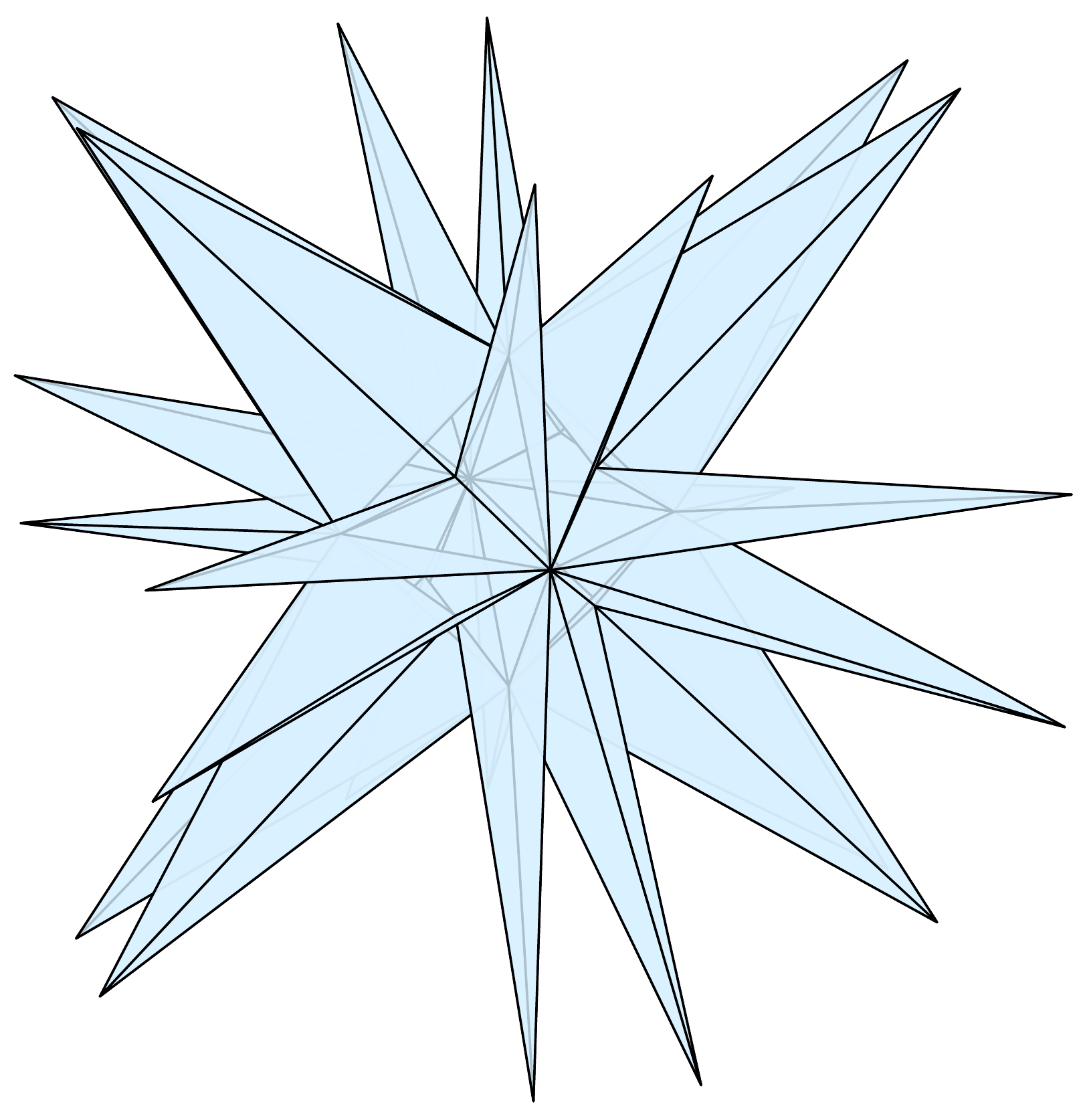}
\caption{Left: Iterated Kleetope $K^2 O$ of an octahedron $O$. Right: Its non-convex isosceles realization.}
\label{fig:kkoct}
\end{figure}

If $G$ is a polyhedral graph, then its \emph{Kleetope} is another polyhedral graph, obtained by adding to~$G$ a vertex for each face of~$P$, adjacent to every vertex of the face.  We denote the Kleetope by $KG$, and we denote the result of repeating the Kleetope operation $i$ times by $K^iG$. For instance, \autoref{fig:kkoct} (left) shows a convex realization of $K^2 O$, where $O=K_{2,2,2}$ is the graph of an octahedron. Geometrically, a realization of a Kleetope can be obtained from a realization of $G$ by attaching a pyramid to each face of $G$. When $G$ is a triangulation, every convex realization of $KG$ leads to a convex realization of $G$, obtained by removing the added vertices of $KG$. (In the non-convex case, this removal may lead to self-intersections. The property that $G$ is a triangulation and not just an arbitrary polyhedral graph is needed to ensure that each pyramid has a flat base.)

One of the earliest applications of Kleetopes was in the construction of triangulations in which all paths or cycles are short. If $G$ is a triangulation with $n$ vertices, then it has $2n-4$ triangles and $KG$ has $3n-4$ vertices. However, in a simple cycle of $KG$, no two vertices of $KG\setminus G$ can be adjacent, and replacing each such vertex by the edge between its two neighbors in $G$ produces a simple cycle of $G$. Therefore, if the longest simple cycle in $G$ has length $c$, then the longest simple cycle in $KG$ has length $2c$. Thus, the cycle length only increases by a factor of two, while the number of vertices increases by a factor converging (for large $n$) to three. Iterating this process produces $n$-vertex triangulations in which the longest path or cycle has length $O(n^{\log_3 2})$~\cite{MooMos-PJM-63}, and iterating it a bounded number of times produces the following result in the form that we need it:

\begin{lemma}
\label{lem:short}
For every $\varepsilon$ there is an $i$ such that, for all polyhedral graphs $G$, the longest simple cycle in $K^i G$
includes a fraction of the vertices of  $K^i G$ that is less than~$\varepsilon$.
\end{lemma}

\begin{proof}
If $G$ is a polyhedral graph, $K G$ is simplicial and has at least 8 vertices. Each subsequent iteration of the Kleetope operation multiplies the number of vertices by at least 5/2 and multiples the longest cycle length by at most 2, so we may take $i$ to be $1+\log_{5/4} 1/\varepsilon$.
\end{proof}

\subsection{Discrete Gauss map}

The \emph{discrete Gauss map} of a polyhedron (in a fixed position in $\mathbb{R}^3$) maps each feature of the polyhedron to a set of three-dimensional unit vectors, the outward-pointing perpendicular vectors of supporting planes of the polyhedron (planes that intersect its boundary but not its interior) that are tangent to the polyhedron at that feature. The set of all unit vectors forms a unit sphere, the \emph{Gauss sphere}, on which each face of the polyhedron is mapped to a single point, each edge of the polyhedron is mapped to an arc of a great circle, and each vertex of the polyhedron is mapped to a convex region bounded by arcs of great circles. That is, the Gauss map is a dimension-reversing map from the features of the given polyhedron to features of a polygonal subdivision of the sphere, sometimes called the spherical dual of the polyhedron~\cite{Sul-Bridges-06}.

We refer to distances along the surface of the Gauss sphere as \emph{geodesic distance}, and the maximum distance between any two points in a given set of unit vectors on the sphere as \emph{geodesic diameter}.
If the Gauss map takes two adjacent faces of the polyhedron to points whose geodesic distance is $d$, then the dihedral angle of the faces is exactly $\pi-d$.

\section{Nonexistence of convex isosceles realizations}

In this section we prove that iterated Kleetopes (with a sufficiently large number of iterations) do not have a convex realization with all faces isosceles. Our proof proceeds in stages, where each stage adds more iterations to the Kleetope and proves the existence of a face whose shape and neighborhood is more strongly constrained, until the constraints are sufficient to rule out isosceles triangles as faces. We will eventually show that in the final stage of the Kleetope construction process, it will be necessary to glue on a triangular pyramid with an obtuse base triangle, isosceles faces other than the base, and  small (very sharp) dihedral angles at the base edges of the pyramid. However, pyramids with these properties do not exist, as the next subsection shows.

\subsection{Isosceles pyramids with obtuse bases}

Every vertex of every convex polyhedron has an incident edge with dihedral angle $>\pi/3$,
but in some polyhedra (for instance flattened pyramids) these large dihedral angles avoid one of the faces. The main result of this section is that in tetrahedra with one face obtuse and the rest isosceles, the obtuse face cannot be avoided by the large dihedrals:

\begin{lemma}
\label{lem:big-dihedral}
Let $T$ be a tetrahedron in which one face $F$ is obtuse, and the other three faces are isosceles.
Then at least one of the dihedral angles on an edge of $F$ is greater than $\pi/3$.
\end{lemma}

\begin{proof}
We divide into cases according to how the base and apex of each isosceles triangle are arranged relative to $F$.
\begin{description}
\begin{figure}[t]
\centering\includegraphics[height=0.38\textwidth]{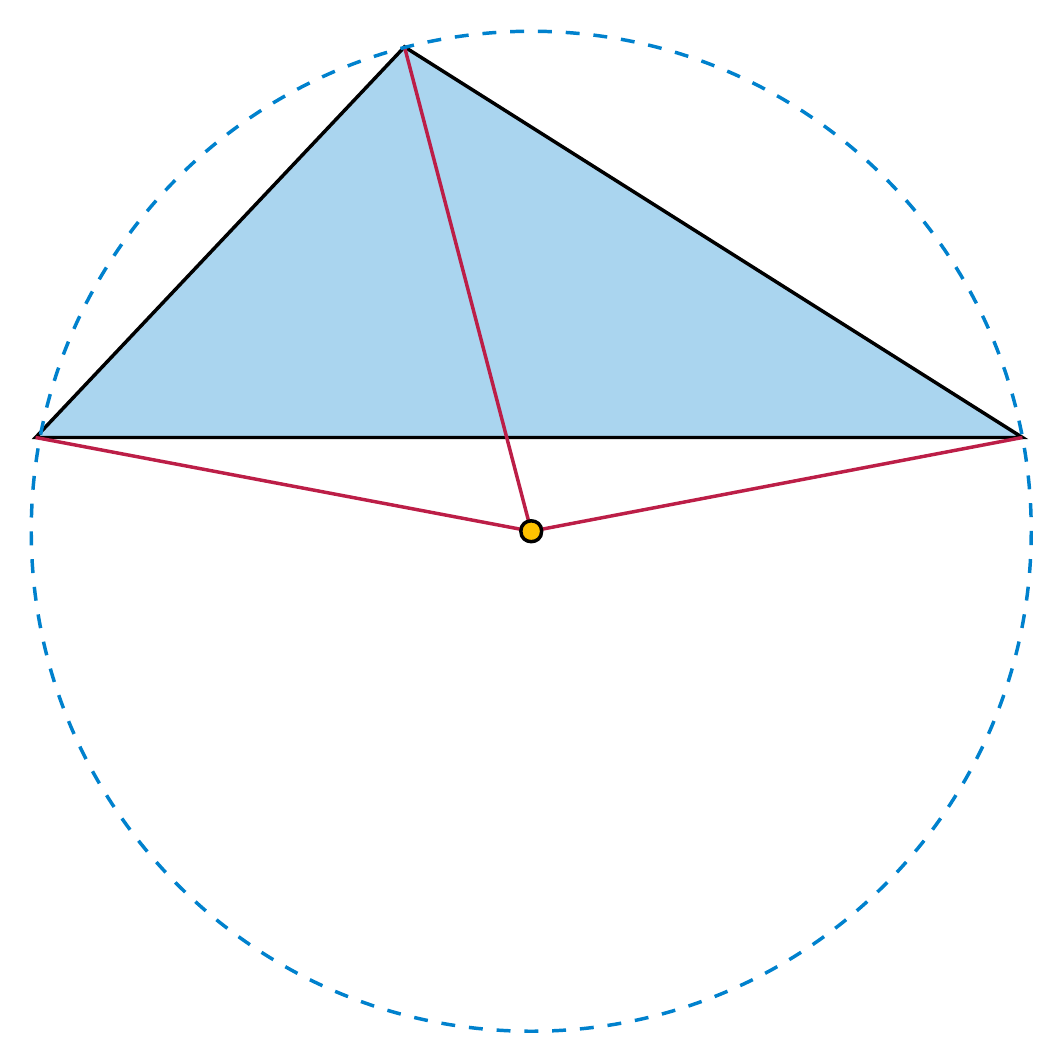}\qquad\raisebox{0.05\textwidth}{\includegraphics[height=0.28\textwidth]{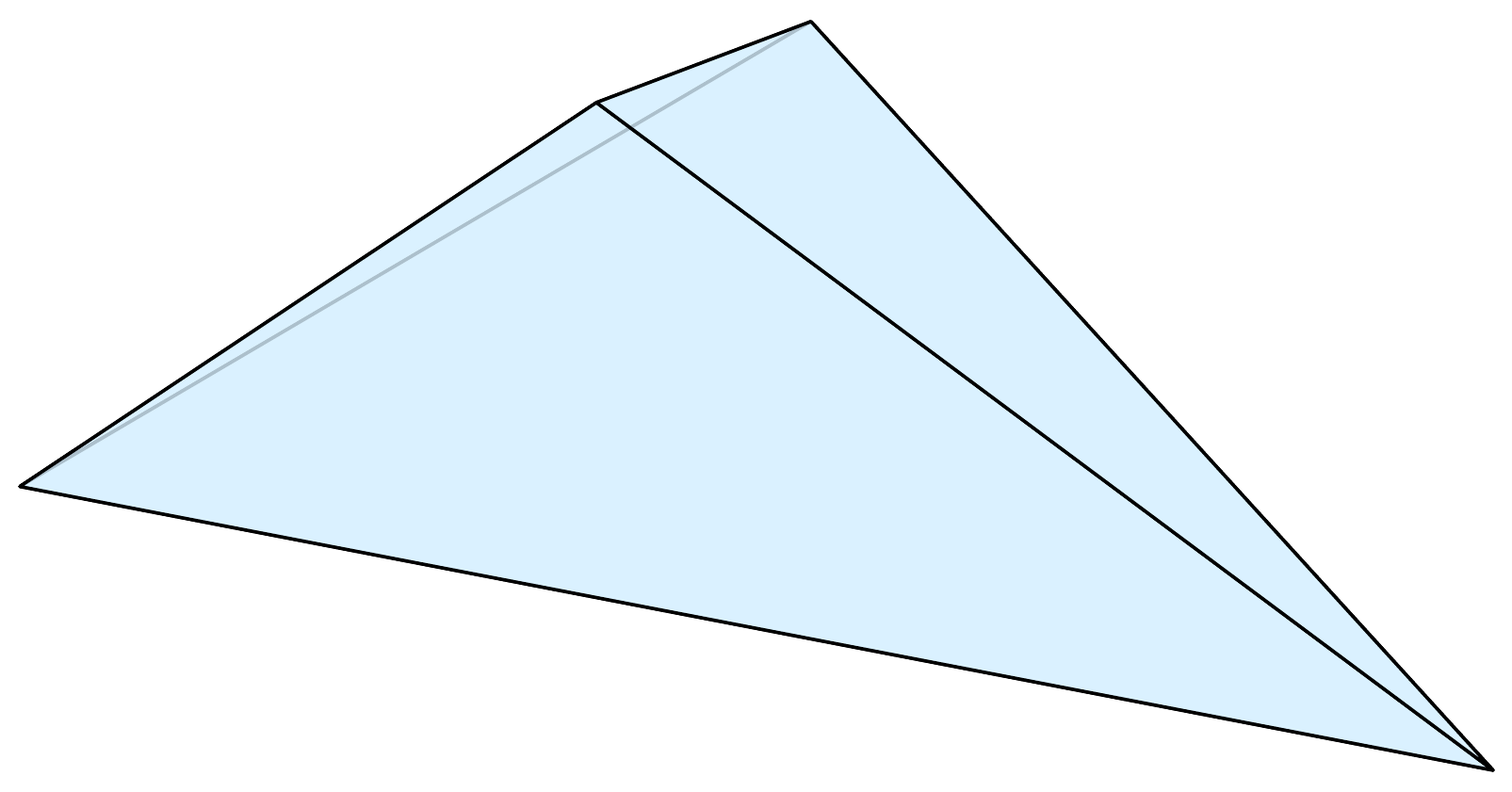}}
\caption{Left: The case of \autoref{lem:big-dihedral} with two bases on $F$, viewed in projection onto the plane containing $F$. The three edges of $T$ disjoint from $f$ (red) have equal lengths, and meet on a line perpendicular to the plane of $F$ through the circumcenter of $F$ (seen in projection as the small yellow circle). Right: The convex hull of two congruent isosceles triangles rotated around a shared edge, used in the case of the lemma with two shared bases and the case with a long side. At an endpoint of the shared edge, either the shared dihedral or the other two equal dihedrals are greater than $\pi/3$.}
\label{fig:big-dihedral-cases}
\end{figure}

\item[Two bases on $F$.] If at least two of the isosceles triangles have their base on an edge of $F$, then the three edges of $T$ that are disjoint from $F$ have equal length, and the vertex of $T$ where these edges meet lies on a line perpendicular to $F$ through the circumcenter of $F$. Because $F$ is obtuse, its long side separates the circumcenter from $F$, so the dihedral of~$T$ on the long side of $F$ is at least $\pi/2$. (See \autoref{fig:big-dihedral-cases}, left. This case requires $f$ to be obtuse.)

\item[Two shared bases.] If two of the three isosceles triangles share a base edge with each other, then in each of them one of the two equal sides is an edge of $F$. The third isosceles triangle has edges of these lengths, equal to the lengths of two edges of $F$, so it is congruent with $F$, and rotated from $F$ around the edge $e$ that it shares with $F$.
\autoref{fig:big-dihedral-cases} (right) depicts the case where the shared edge is the base of the isosceles triangle, but it could also be one of the sides. Because the two triangles on $e$ are congruent, the vertex figure of $T$ at either endpoint of $e$ (the intersection of $T$ with a sufficiently small sphere centered on that point) is a spherical isosceles triangle. Like any spherical triangle its angles sum to greater than $\pi$ so one angle is at least $\pi/3$. This large angle is either the dihedral at $e$ (the apex of the spherical isosceles triangle) or is the equal dihedral at the other two edges meeting at the same vertex of $T$, one of which is an edge of~$F$. (This case does not require $f$ to be obtuse, and can produce dihedral angles arbitrarily close to $\pi/3$.)

\item[Long side.] If the long edge $e$ of $F$ is one of the two equal sides of one of the isosceles triangles, consider the tetrahedron $T'$ formed as the convex hull of $F$ and a rotated copy of $F$ by an angle of $\pi/3$ around $e$ (on the same side of the plane through $F$ as $T$). As in the case of two shared bases, the vertex figures at the endpoints of $e$ are spherical isosceles triangles, and the choice of $\pi/3$ as rotation angle implies that, in $T'$, all the dihedrals at the endpoints of $e$ (which include all three dihedrals on the edges of $f$) are at least $\pi/3$.

The two endpoints of $e$ are the unique farthest pair of points in $T'$, so the edge of $T$ of equal length to $e$ must end at a point outside $T'$. This point is separated from $T'$ by one of the face planes of $T'$, which meets $F$ at one of its edges. The dihedral of $T$ on that edge is greater than the dihedral of $T'$, which is in turn at least $\pi/3$.

\begin{figure}[t]
\centering\includegraphics[scale=0.5]{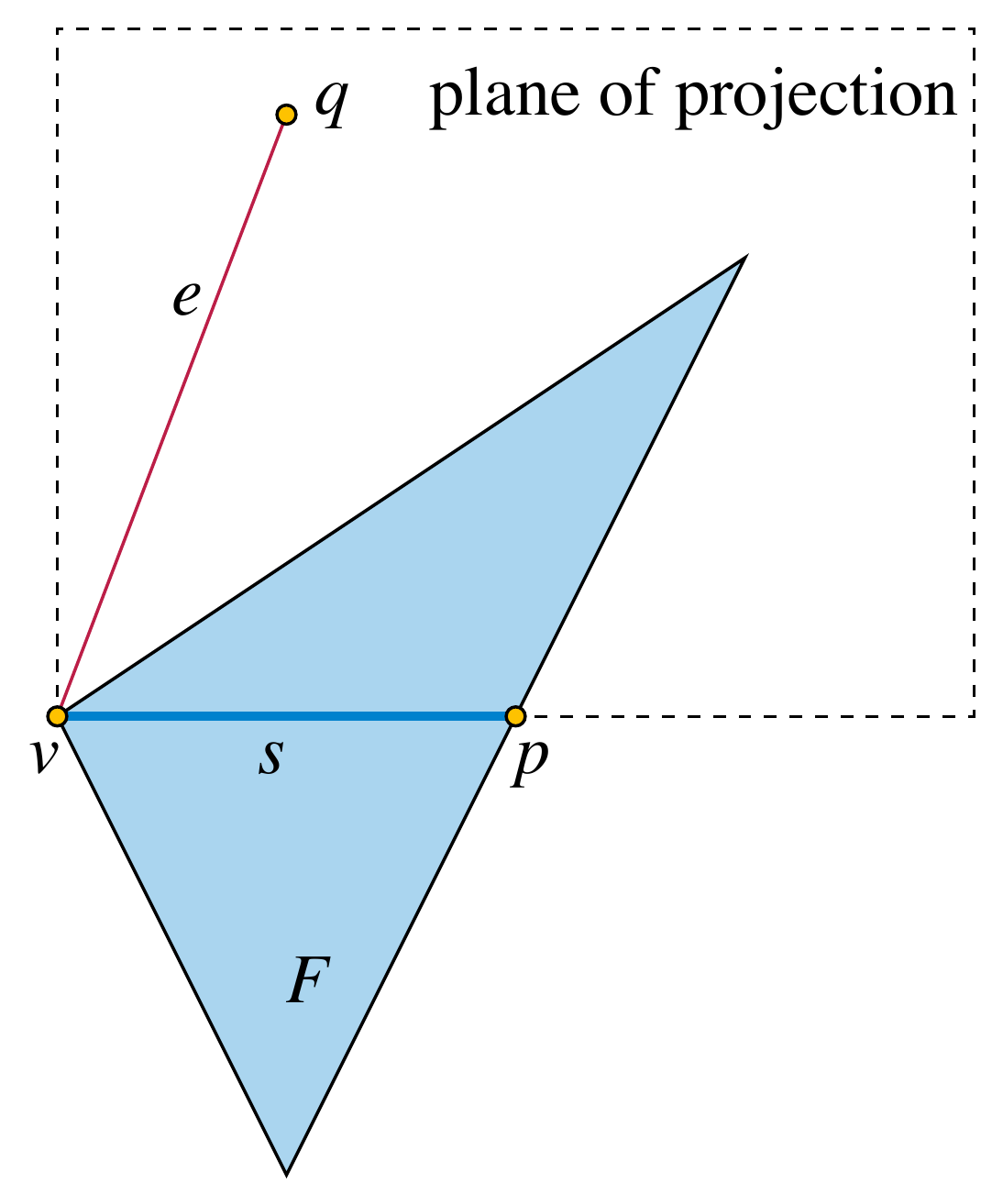}
\qquad\raisebox{8ex}{\includegraphics[scale=0.5]{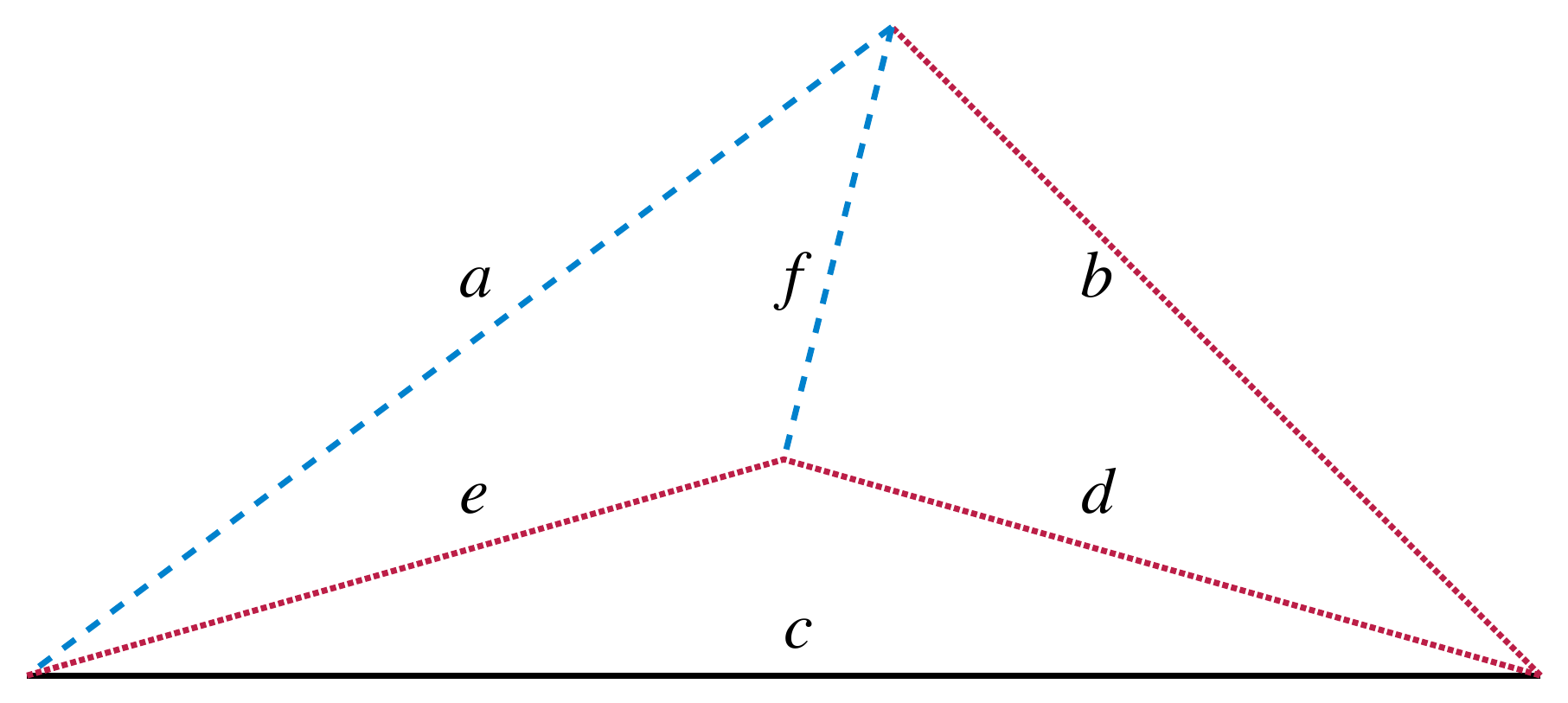}}
\caption{Left: Notation for the case of \autoref{lem:big-dihedral} of two shared apexes.
Right: Notation for the final ``none of the above'' case, shown schematically with equal-length edges in equal colors.}
\label{fig:last-cases}
\end{figure}

\item[Two shared apexes.] If two of the isosceles faces of $T$ both have their apex at the same vertex $v$ of $F$, then the three edges of $T$ meeting at that vertex all have the same length $\ell$. As two of them are edges of $F$, $F$ must be isosceles, with these edges as its sides. We can view this case in projection onto a plane perpendicular to $F$ through its center line. Of the three equal shared edges, the edge $e$ that is disjoint from $F$ lies in this plane, while the other two lie in symmetric positions on either side of it, and project to a point $p$ at the midpoint of the base of $F$. Triangle $F$ itself projects to a line segment $s$ connecting $v$ to $p$. Let $q$ be the non-shared endpoint of edge $e$ (\autoref{fig:last-cases}, left).

Then, in triangle $vpq$, edge $vq$ is longer than edge $vp$, because $vp$ is the projection of a three-dimensional line segment of equal length to $vq$. Let $\theta=\max(\angle qvp,\angle qpv)$, the larger of the two angles on side $s$ of this triangle. Then $\theta\ge\theta^*$ where $\theta^*$ is the base angle of the isosceles triangle obtained by moving $q$ to lie on the perpendicular bisector of $s$ (keeping the length of edge $vp$ the same). Because $vq$ is longer than $vp$, $\theta\ge\theta^*>\pi/3$. If $\theta=\angle qpv$, it is also the dihedral angle on the base of $f$, which is therefore greater than $\pi/3$. If $\theta=\angle qvp$, it is less than the dihedral angles on the two sides of $f$ (because these sides are not perpendicular to the plane of projection), which are therefore both greater than $\pi/3$. (This case does not require $F$ to be obtuse.)

\item[None of the above.] In the remaining case, let $a$, $b$, and $c$ be the sides of $F$, with $c$ being the longest side, and let $d$, $e$, and $f$ be the other three sides of $T$, with $d$ opposite $a$, $e$ opposite $b$, and $f$ opposite $c$. In order to avoid the case of a long side, isosceles triangle $cde$ must have $c$ as its base, so $d$ and $e$ must have equal length.
In order to avoid the cases of a shared base or a shared basis, exactly one of the other two triangles must have $d$ or $e$ as its base; by symmetry, we may assume without loss of generality that isosceles triangle $aef$ has $e$ as its base (so $a$ and $f$ are equal) and that isosceles triangle $bdf$ has $f$ as its base (so $b$, $d$, and $e$ are equal). \autoref{fig:last-cases} shows this case schematically (not to scale), with edges that in $T$ have equal lengths depicted as having equal colors and textures. We divide into sub-cases according to the comparison between these lengths:

\begin{itemize}
\item If $a$ and $f$ have the same length as $b$, $d$, and $e$, then we would be in the case of two shared bases (shared at $f$).
\item If $a$ and $f$ are shorter than $b$, $d$, and $e$, then consider the tetrahedron $T'$ formed from $T$ by lengthening $a$ and $f$ to have the same length as $b$, $d$, and $e$, but keeping all other edge lengths the same. $T'$ has two equilateral triangle faces $aef$ and $bdf$, and two isosceles faces with apex angle greater than equilateral, $abc$ and $cde$ (although these faces need not be obtuse). The same analysis as the case of two shared bases shows that, in $T'$, the dihedrals on $a$, $b$, $d$, and $e$ are all greater than $\pi/3$. Lengthening $a$ and $f$ causes the dihedral on $b$ to decrease (among other changes), so in $T$ as well, $b$ is greater than $\pi/3$.
\item If $a$ and $f$ are longer than $b$, $d$, and $e$, then consider the tetrahedron $T'$ formed from $T$ by lengthening $b$ to equal $a$ and $f$, keeping all other edge lengths the same. Then $T'$ falls into the case of two shared apexes, although in $T'$ triangle $abc$ may not be obtuse. From the analysis of that case, $T'$ has a dihedral greater than $\pi/3$ on at least one of $a$ and $c$. Lengthening $b$ causes these two dihedrals to decrease, so in $T$ as well, at least one of $a$ and $c$ has a dihedral greater than $\pi/3$.\qedhere
\end{itemize}
\end{description}
\end{proof}

\subsection{Kleetopes have faces without sharp dihedrals}

We define the \emph{sharpness} of a face $F$ of a polyhedron to be $\pi-\theta$ where $\theta$ is the smallest dihedral angle of an edge of $F$. Equivalently, on the Gaussian sphere, the sharpness of $F$ is the length of the longest edge connecting the Gaussian image of $F$ to one of its neighbors. In a face with small sharpness, all dihedrals are close to $\pi$. In terms of this quantity, \autoref{lem:big-dihedral} from the previous section can be interpreted as stating that one cannot glue an isosceles pyramid onto an obtuse triangle of sharpness less than $\pi/3$ while preserving the convexity of the resulting polyhedron.

It is possible for a convex polyhedron to have arbitrarily many faces which all have large sharpness; consider, for instance, the realizations of bipyramids (\autoref{fig:bipyramid}, left). As we show now, this is not true of sufficiently-iterated Kleetopes.

\begin{lemma}
\label{lem:unsharp}
For every $\varepsilon$ there exists $i$ such that, for all polyhedral graphs $G$ and all convex realizations of $K^i G$, the realization includes at least one face whose sharpness is $\le\varepsilon$.
\end{lemma}

\begin{proof}
Construct an arrangement of $c=O(1/\varepsilon)$ great circles on the Gauss sphere that partition the sphere into cells with geodesic diameter at most $\epsilon$, for instance by choosing two points on the sphere at geodesic distance $\pi/2$ and choosing for both points a collection of great circles crossing each other at equal angles at that point. By  \autoref{lem:short}, let $i$ be such that the longest cycle in $K^i G$ includes a fraction less than $1/3c$ of the edges of $K^i G$. We claim that $i$ has the stated property.

To prove this, consider any given choice of $G$ and convex realization of $K^i G$, and rotate the arrangement of great circles (if necessary) so that none of them passes through a vertex of the spherical dual of $K^i G$.
Then each of the great circles passes through a sequence of edges and faces of the spherical dual corresponding to a simple cycle in $K^i G$, the cycle formed as the silhouette of~$K^i G$ in a perpendicular projection onto the plane of the great circle. By \autoref{lem:short} and the choice of~$i$, these silhouette cycles together use a fraction less than $1/3$ of the edges of $K^i G$. It follows that at least one face $F$ of $K^i G$ is bounded by edges that are not in any silhouette cycle, for if all faces had at least one edge in a cycle the total fraction of edges in cycles would be at least $1/3$. The Gaussian images of $F$ and its three neighboring triangles in $K^i G$ are all in the same cell of the arrangement of great circles, of geodesic diameter at most $\varepsilon$. It follows that the edge from the Gaussian image of $F$ to  the farthest of its neighbors has length at most $\varepsilon$.
\end{proof}

If we glue a pyramid onto a face $F$ of a convex polyhedron to produce another convex polyhedron (as happens for realizations of Kleetopes), the faces of the pyramid that replace $F$ can have greater sharpness than $F$ had. However, the sharpness cannot increase too much:

\begin{lemma}
\label{lem:doublesharp}
Let $P$ be a triangulation, let $Q$ be a convex realization of $K^i P$ for some $i$, and let $F$ be a face of $P$ that (in the realization of $P$ derived from $Q$ by removing the vertices of $Q\setminus P$) has sharpness $\varepsilon$. Then every face of $Q$ formed by  subdividing $F$ has sharpness at most $2\varepsilon$.
\end{lemma}

\begin{proof}
On the Gaussian sphere, the image of $F$ is a point, incident to three great circle arcs of length at most $\varepsilon$ corresponding to the dihedrals along the edges of $F$. The spherical triangle $T$ forming the convex hull of these three arcs has geodesic diameter at most $2\varepsilon$, by the triangle inequality. Subdividing $F$ to form $Q$ replaces the Gaussian image of $F$ (a single vertex in the spherical dual of $P$) by a subgraph, attached to the rest of the graph at three vertices on the same three arcs. The sharpness of any face of the subdivision equals the length of the longest edge incident to its Gaussian image, which may either be an edge of this subgraph or one of the three edges connecting this subgraph to the rest of the graph. All of these edges lie within spherical triangle $T$, so their lengths are at most the diameter of $T$, which is at most $2\varepsilon$.  
\end{proof}

\subsection{Kleetopes have obtuse faces}

The tetrahedron formed as the convex hull of the origin and the three perpendicular unit basis vectors of three-dimensional space has the same shape as the truncated corner of a cube, with three isosceles right triangle faces and one equilateral face.
When the right triangle faces are projected perpendicularly onto the equilateral face, their projections are isosceles triangles with apex angle $2\pi/3$ and base angles $\pi/6$. The dihedral angles of this tetrahedron are $\pi/2$ between pairs of right-triangle faces, and $\tfrac{1}{2}\cos^{-1}(-\tfrac{1}{3})\approx 0.304\,\pi$ between the right triangles and the equilateral triangles; we use $\varphi$ to denote this latter angle. The familiar geometry of this example forms the extreme case for the following lemma:

\begin{lemma}
\label{lem:obtuse-projection}
Let $T$ be a triangle in $\mathbb{R}^3$, let $e$ be an edge of $T$, and $P$ be a plane through $e$ forming an angle less than $\varphi$ with the plane of $T$. Suppose that the perpendicular projection $T'$ of $T$ onto $P$ is a triangle $T'$ whose largest angle
is at least $2\pi/3$, at the vertex opposite $e$. Then $T$ is obtuse.
\end{lemma}

\begin{proof}
Increasing the angle between the plane of $T$ and $P$ can only increase the largest angle of~$T'$, so we may assume without loss of generality that the angle between the plane of $T$ and $P$ is exactly $\varphi$, and prove that in this case $T$ is either right or obtuse. In the plane of $T$, the points $p$ that form an acute vertex of a triangle with opposite edge $e$ are separated from the points that form an obtuse vertex by a semicircle having $e$ as its diameter.
The midpoint of this semicircle is a point that, together with $e$, forms an isosceles right triangle, and its projection onto $P$ is the center point $c$ of an equilateral triangle that has $e$ as one of its three sides. The projection of the semicircle itself onto $P$ is a semi-ellipse having $e$ as its diameter and passing through~$c$ (the boundary of the outer yellow region in \autoref{fig:obtuse-projection}).

\begin{figure}[t]
\centering\includegraphics[scale=0.5]{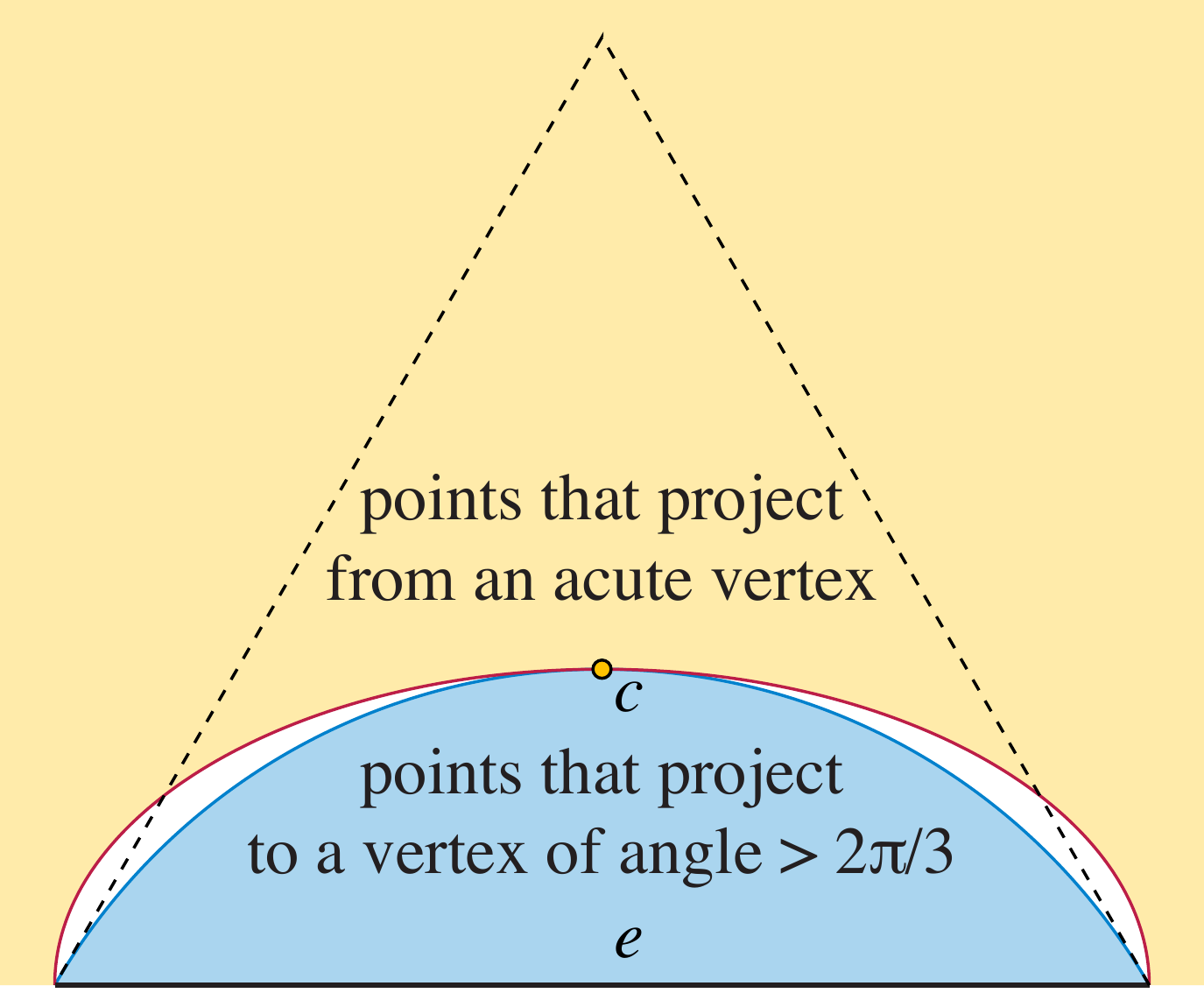}
\caption{A semi-ellipse (projected from a semicircle) and a circular arc through the same three points (endpoints of $e$ and the center $c$ of an equilateral triangle on edge $e$) separate points that project from an acute angle and points that project to an angle $>2\pi/3$ in the proof of \autoref{lem:obtuse-projection}.}
\label{fig:obtuse-projection}
\end{figure}

In $P$, the points $q$ that form a vertex of angle $>2\pi/3$ triangle with opposite side $e$ are separated from the points that form a vertex of angle $<2\pi/3$ by a circular arc of angle $\pi/3$ connecting the two endpoints of $e$ through $c$ (the boundary of the inner blue region in \autoref{fig:obtuse-projection}). This circular arc lies inside the semi-ellipse, except for the endpoints of $e$ (where the circle and ellipse, if continued, would cross) and the point $c$ where the two curves are tangent (counting as a double contact and using up all four possible points of intersection of a circle and ellipse). For $T'$ to have angle at least $2\pi/3$ opposite $e$, its vertex must be on or inside the circular arc, and therefore either inside the semi-ellipse or on it at $c$. To project to a point on or inside the semi-ellipse, the vertex of $T$ opposite~$v$ must be on or inside the semi-circle, and therefore $T$ must be obtuse or right.
\end{proof}

Using this lemma, we can prove that sufficiently iterated Kleetopes have an obtuse face, and, moreover, one with low sharpness:

\begin{lemma}
\label{lem:has-obtuse}
Let $\varepsilon$ be a positive number with $\varepsilon<\varphi$, let $i$ be an integer large enough that, by \autoref{lem:unsharp}, every convex realization of $K^i P$ has a face of sharpness less than $\epsilon$, and let $P$ be any triangulation. Then every convex realization of $K^{i+1} P$ has an obtuse-triangle face of sharpness less than $2\varepsilon$.
\end{lemma}

\begin{proof}
Given a realization of $K^{i+1} P$, form from it a realization of $K^i P$ by removing the extra vertices added in going from $K^i P$ to $K^{i+1} P$, and apply \autoref{lem:unsharp} to this realization to find a face $F$ of $K^i P$ whose sharpness is less than $\varepsilon$. Now consider the three faces that subdivide $F$ in $K^{i+1} P$, and project these three faces perpendicularly onto the plane of $F$. Their projections must be contained within $F$, for otherwise the angle between one of these faces and $F$ would be greater than $\pi/2$, too big to form a convex dihedral with the neighboring face across the edge of $F$ given the low sharpness of $F$ itself. Therefore, within the plane of projection, these three projected triangles subdivide $F$ and meet at a shared vertex interior to $F$. If $F'$ is the face of $K^{i+1}$ whose projection has the largest angle at this shared vertex, then this angle is least $2\pi/3$. By the assumption on the sharpness of $F$, the plane of $F'$ makes an angle less than $\varphi$ with the plane of $F$, so we can apply  \autoref{lem:obtuse-projection} to $F'$ and its projection. By this lemma, $F'$ is obtuse, and by \autoref{lem:doublesharp} its sharpness is less than $2\varepsilon$.
\end{proof}

\subsection{Iterated Kleetopes cannot be convex and iscosceles}

Combining the results from this section, we have:

\begin{theorem}
There exists an integer $j$ such that, for every triangulation $P$, every convex realization of $K^j P$ has a non-isosceles face.
\end{theorem}

\begin{proof}
Let $i$ be an integer large enough that, by \autoref{lem:unsharp}, every convex realization of $K^i P$ has a face of sharpness less than $\pi/6$, and let $j=i+2$. Let $P$ be an arbitrary triangulation. Consider any convex realization of $K^j P$, and the convex realization of $K^{j-1} P=K^{i+1} P$ obtained from it by deleting the last layer of added vertices. By \autoref{lem:has-obtuse}, the realization of $K^{j-1} P$ has an obtuse face $F$ of sharpness less than $\pi/3$. If the realization of $K^j P$ had all-isosceles faces, \autoref{lem:big-dihedral} (applied to the tetrahedron formed by $F$ and the three isosceles triangles that subdivide it) would show that one of these isosceles triangles lies in a plane making an angle greater than $\pi/3$ with the plane of~$F$. Because this angle is greater than the flatness of $F$, it is impossible for such a triangle to make a convex dihedral angle with whatever other triangle neighbors it across the edge it shares with~$F$, contradicting the assumption that the realization is convex. This contradiction shows that the assumption that the realization has only isosceles triangles as its faces cannot be true.
\end{proof}

\section{Non-convex realization of Kleetopes}
 
In this section we prove that, whenever $G$ is a triangulation, $KG$ has a non-convex realization in which all faces are isosceles. In particular, this is true of the polyhedra that we we constructed as examples that have no convex isosceles realization.

\begin{theorem}
Let $P$ be a convex polyhedron whose graph $G$ is a triangulation. Choose distance~$R$, for each face $f$ of $P$ place a new vertex  at distance $R$ from each of the three vertices of $f$, separated from $P$ by the plane of $f$, and replace $f$ by three triangles through the new vertex and two of the three vertices of $f$. Then for sufficiently large values of $R$, the result is a non-self-crossing realization of $KG$, in which all faces are isosceles triangles.
\end{theorem}

\begin{proof}
Given any triangle $f$ in three-dimensional space, two points at distance $R$ from the three vertices of $f$ will exist whenever $R$ is greater than the circumradius $r_f$ of $f$. These two points are symmetric from each other across the plane of $f$, and the constraint that this plane separates each new point from $P$ uniquely determines the placement used in the theorem. Therefore, whenever $R$ is greater than $\max_f r_f$, the construction given in the statement of the theorem is well-defined. Each new face is a triangle with two sides of length $R$, so the isosceles shape of the faces is immediate. The nontrivial part of the theorem is the claim that the resulting system of triangles is non-self-crossing.

The new point $p$ for $f$ lies on a line through the circumcenter of $f$, perpendicular to $f$. The circumcenter has distance $\le r_f$ from each edge of $f$, and the new point is at distance at least $R-r_f$ away from the plane of $f$, by the triangle inequality. It follows that, although the dihedral angle between $f$ and each triangle incident to $p$ may be obtuse (as happens when $f$ itself is obtuse), it may not be very obtuse: it is $\pi/2+O(r_f/R)$. By choosing $R$ sufficiently large, we may make this angle as close to $\pi/2$ as we desire.

We observe that the Voronoi diagram of the faces of $P$ (the partition of the space outside $P$ into cells within which one face is nearer than all others) has cells bounded by the planes bisecting the dihedral angles of $P$. Let $\theta<\pi$ be the largest dihedral angle of $P$. Then, as long as we choose $R$ large enough that each dihedral angle between a new triangle and the face $f$ that it replaces is less than $\pi-\theta/2$, each new point and its incident triangles will lie within the Voronoi cell of the face that it replaces. Since these Voronoi cells are disjoint, the new triangles cannot cross each other.
\end{proof}

\autoref{fig:kkoct} (right) depicts a non-convex isosceles realization of the iterated Kleetope of the octahedron, obtained by this construction.

\section{Congruent isosceles faces}

\begin{figure}[t]
\centering\includegraphics[width=0.3\textwidth]{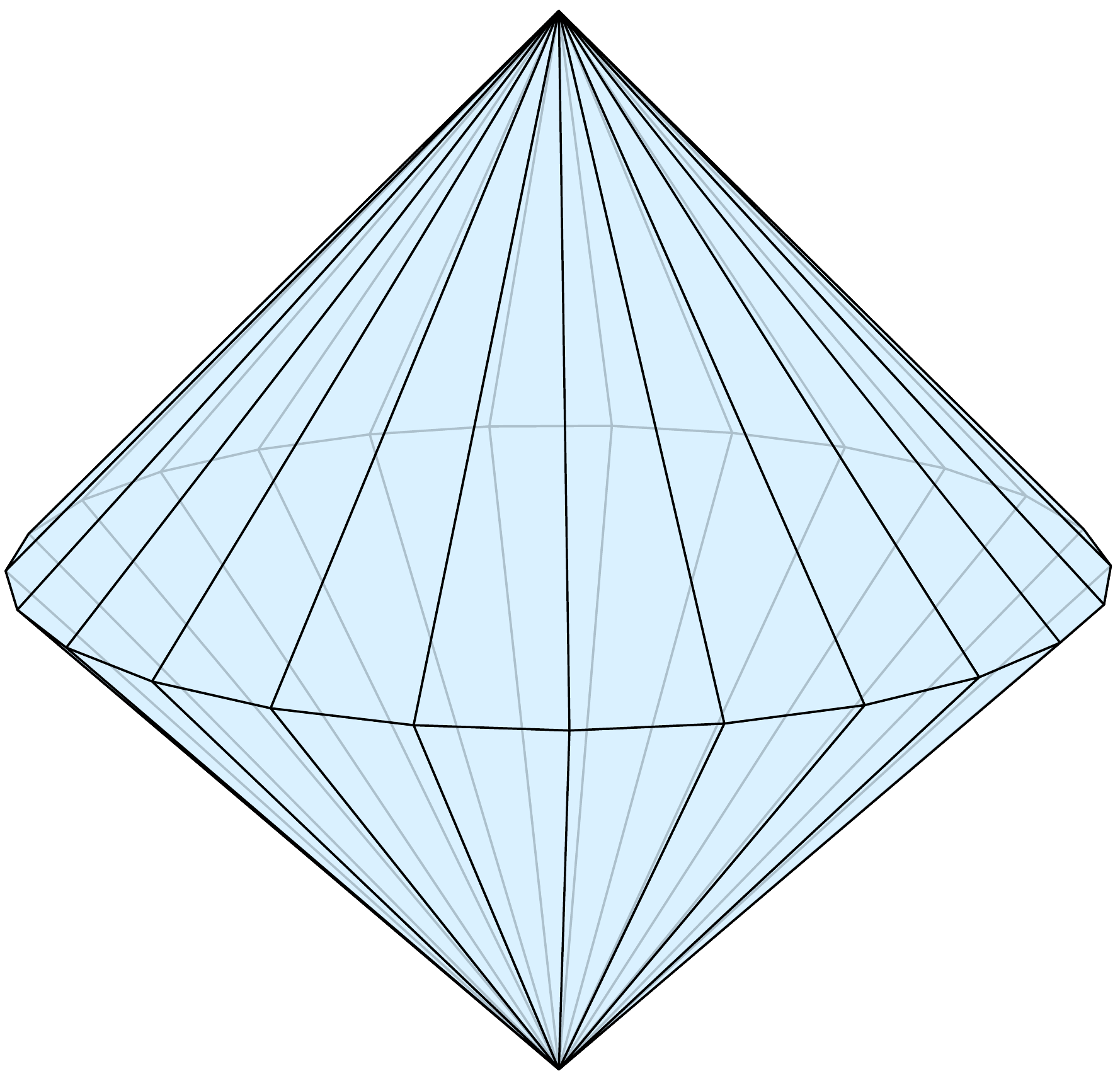}\qquad\includegraphics[width=0.65\textwidth]{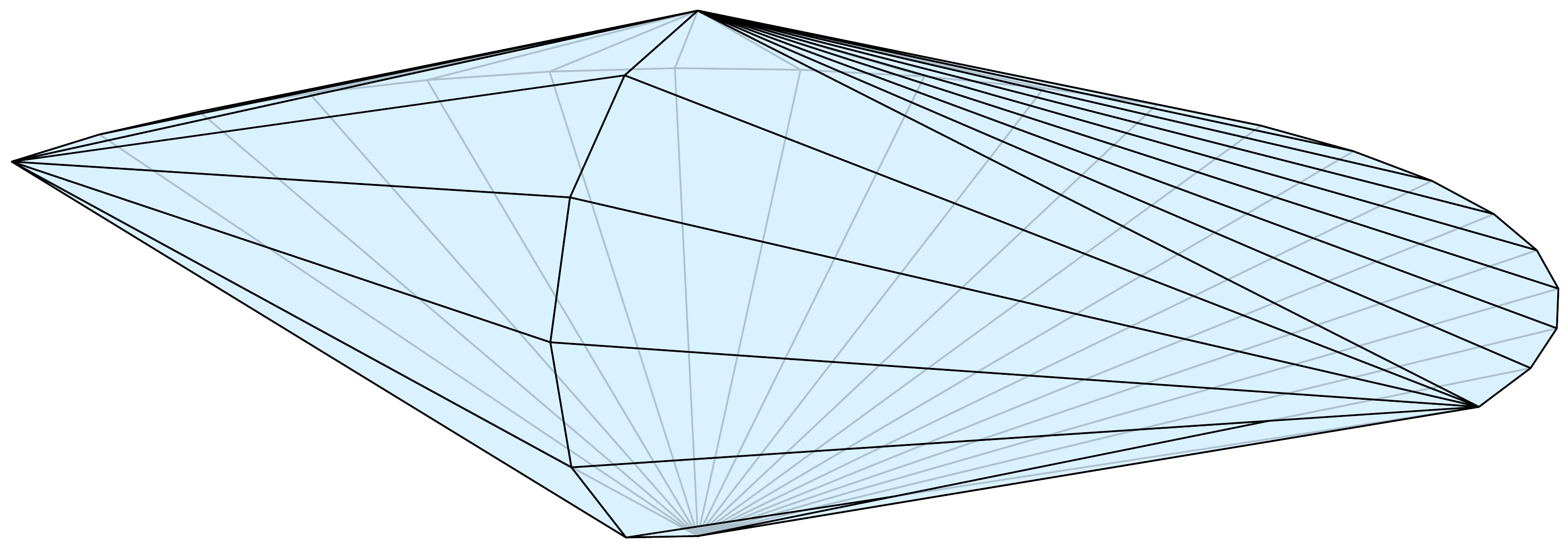}
\caption{Left: Bipyramid. Right: Biarc hull.}
\label{fig:bipyramid}
\end{figure}

\begin{figure}[t]
\centering\includegraphics[width=0.45\textwidth]{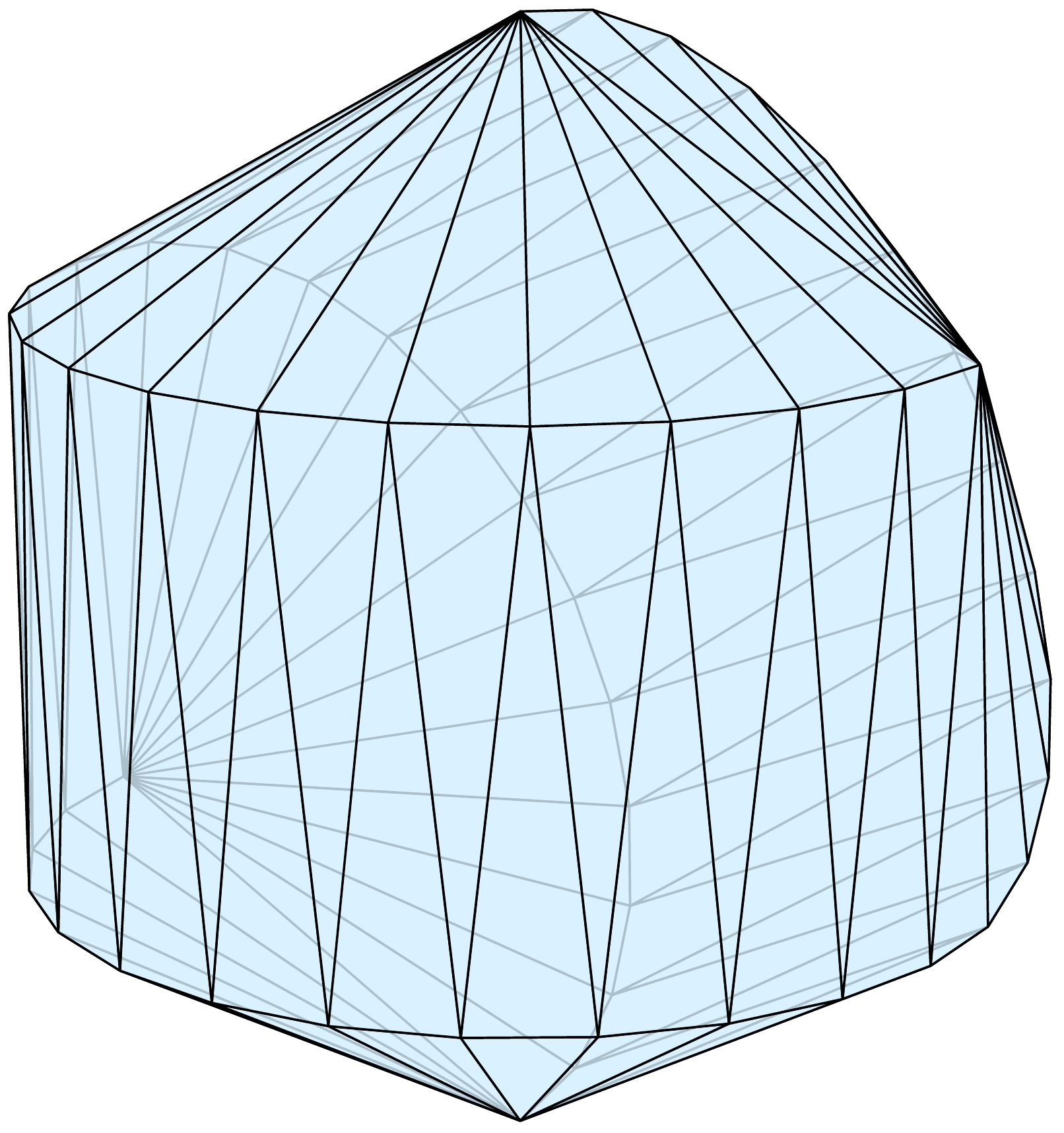}
\caption{An order-25 twisted gyroelongated bipyramid.}
\label{fig:twisted}
\end{figure}

There are only eight convex polyhedra with congruent equilateral triangle faces, called \emph{deltahedra}, but infinitely many non-convex deltahedra~\cite{FreVDW-SS-47,Tri-MM-78}. Generalizing this concept, Malkevitch~\cite{Mal-COMAP-19} asks about convex polyhedra with congruent isosceles triangle faces. He states that there are three known infinite families of these monohedral isosceles convex polyhedra (leaving as a puzzle for the reader finding these families) and poses as an open question whether there are any more. We answer this question by identifying four families:

\begin{itemize}
\item The \emph{bipyramids} are convex hulls of equally spaced points on a circle, together with two additional points on the axis through the center of the circle, equally far from it on each side (\autoref{fig:bipyramid}, left).
\item Let $X$ and $Y$ be two arcs of circles, each equidistant from the endpoints of the other arc, such that both arcs lie on the boundary of the convex hull of their union. Then the convex hull of any finite set of points on both arcs that includes the endpoints of each arc consists of equilateral triangles whose base connects two consecutive points on one arc and whose apex is on the other arc. For any two integers $x\ge 2$ and $y\ge 2$ it is possible to choose two arcs, such that $x+1$ equally-spaced points on arc $X$ and $y+1$ equally-spaced points on arc $Y$ all have the same distance between consecutive points; for instance, as in \autoref{fig:bipyramid} (right), this can be done by selecting $X$ and $Y$ to be concentric semicircles on perpendicular planes with a carefully-chosen ratio of radii. The convex hull of the resulting point set has $2(x+y)$ congruent isosceles triangles as faces; we call this shape a \emph{biarc hull}.
\item Goldberg~\cite{Gol-AMM-36} constructed an infinite family of convex polyhedra with congruent isosceles faces by attaching pyramids to opposite faces of an antiprism (\autoref{fig:grunbaum}, left). For an antiprism over a $k$-gon, the result has $4k$ faces. The cases $k=4$ and $k=5$ can also be realized by equilateral triangle faces, producing a gyroelongated square bipyramid and regular icosahedron respectively. As the name ``Goldberg polyhedra'' already refers to an unrelated class of polyhedra, we call the general case an \emph{order-$k$ gyroelongated bipyramid}.
\item In a gyroelongated bipyramids of odd order $k$, the faces can be partitioned into two subsets, separating each pair of opposite faces. For certain realizations of these polyhedra, the edges separating these two subsets form a non-planar hexagon with order-6 symmetry. Twisting one of the two subsets by an angle of $2\pi/3$ leaves this hexagon in the same position and results in another polyhedron, which we call the order-$k$ twisted gyroelongated bipyramid (\autoref{fig:twisted}). We detail the construction below.
\end{itemize}

\subsection{Twisted gyroelongated bipyramids}

To understand the twisted gyroelongated bipyramids, it is helpful to into a little more detail about Goldberg's method for realizing order-$k$ gyroelongated bipyramids in such a way that all faces are congruent and isosceles. In Goldberg's construction, all of the points lie on a unit sphere, with the two pyramid apexes at its poles and the other points spaced equally around two circles of equal and opposite latitude. Goldberg chooses these latitudes so that the spherical angles of the great-circle arcs on the sphere, connecting the points in the same pattern as the polyhedron, are $2\pi/k$ at the apex of each triangle and $(1/2-1/2k)\pi$ at each base angle. In this way, $k$ of these spherical triangles meet with total angle $2\pi$ at each pole, and the five triangles meeting at each remaining point have total angle $(2/k+4(1/2-1/2k))\pi=2\pi$. The angular excess of each triangle, the amount by which its total angle exceeds $\pi$, is $\pi/k$; for spherical triangles, the angular excess equals the area, so the area of each triangle is exactly $1/4k$ times the surface area of the entire sphere. The subdivision of the sphere into spherical triangles of this shape can be transformed into a polyhedron by taking the convex hull of its vertices, and (as is true whenever a spherical triangulation consists only of acute triangles) the combinatorial structure of the resulting convex hull is the same as the spherical triangulation. In this way, one obtains a gyroelongated bipyramid with congruent isosceles faces that is, moreover, both inscribed in the unit sphere and circumscribes another smaller sphere. The polyhedron illustrated in \autoref{fig:grunbaum} (left) was constructed in this way.

Now, divide the resulting polyhedron into two polyhedral surfaces by a non-planar hexagon through both poles, separating each face from its opposite face. The same subdivision into two surfaces can also be made on the spherical triangulation from which the polyhedron was derived. On the surface of the sphere, at each pole, the hexagon subdivides $(k-1)/2$ spherical triangles on one side from $(k+1)/2$ spherical triangles on the other side, so its two edges make an angle of $(1-1/k)\pi$ with respect to each other at the poles. At each non-pole vertex, the hexagon subdivides two spherical triangles (meeting at their base) on one side from three spherical triangles (at two base angles and an apex) on the other, so its angle is $2(1/2-1/2k)\pi=(1-1/k)\pi$. That is, this hexagon lies on a sphere, has equal edge lengths, has equal angles, and splits the sphere into two equal areas. It follows that it has order-6 symmetry: it can be twisted by an angle of $2\pi/3$ or $4\pi/3$, or flipped and then twisted, to produce the same hexagon. Applying the twist by an angle of $2\pi/3$ to half of the triangles in Goldberg's spherical triangulation (leaving the other half fixed) produces another spherical triangulation with congruent isosceles spherical triangles, which was described in 2005 by Robert Dawson~\cite{Daw-RB-05}.

Dawson's spherical triangulation can be transformed into a monohedral isosceles polyhedron by taking the convex hull of its vertices as in Goldberg's construction. This produces a new polyhedron, the twisted gyroelongated bipyramid, for every odd $k>5$. For $k=3$ it produces a triangulation of a cube by some of the diagonals of its squares (with all faces right and isosceles, but not forming a strictly convex polyhedron) and for $k=5$ the twisting process leaves the regular icosahedron unchanged. In relation to the work of Dawson, we remark that (for any $n\ge 2$) it is also possible to tile the sphere by $4n$ isosceles spherical triangles with base angle $\pi/n$ and apex angle $\pi(1-1/n)$, meeting edge-to-edge in the pattern of a bipyramid. (This is the spherical analog of the non-convex polyhedron depicted in \autoref{fig:nwb}.) When $n$ is odd, the resulting spherical triangulation again has a skew hexagon with order-6 symmetry, and twisting one side of this skew hexagon by an angle of $2\pi/3$ leads to another spherical triangulation with congruent isosceles spherical triangles, not listed by Dawson. However, this construction does not directly lead to a convex polyhedron, because the spherical triangles are not obtuse, and in \autoref{thm:nwb-sharp} we will see more strongly that (with finitely many exceptions) there is no convex polyhedron with isosceles triangles in this pattern.

We note that removing a small number of strips of four triangles, from pole to pole in one of the two halves of the twisted gyroelongated bipyramid, but preserving the geometric shape of the remaining triangles, results in a polyhedral surface that is less symmetric (there are different numbers of triangles on one of the two halves than the other) but still obeys the conditions of Alexandrov's uniqueness theorem characterizing the metric spaces derived from the surfaces of polyhedra~\cite{Ale-05}: it is topologically spherical and locally Euclidean except at a finite number of points of positive angular defect, with total angular defect $4\pi$. It follows from Alexandrov's theorem that the result is the surface of a convex polyhedron, but the theorem does not tell us whether the edges of the polyhedron follow the edges of the triangles from which the surface was derived, or are rearranged into a different combinatorial structure. If the edges of the resulting polyhedron coincide with the edges of the triangles, the result would be a generalized asymmetric form of the twisted gyroelongated bipyramid, with different numbers of congruent isosceles triangles on each of its two sides. We leave as open for future research whether this generalized twisted gyroelongated bipyramid actually exists.

\subsection{Sharp angles}

The known infinite families of monohedral isosceles convex polyhedra all have very sharp apex angles when they have large numbers of faces. This is not a coincidence: as we show, in such a polyhedron with $n$ faces, the apex angle must be $O(1/n)$.

Our proof depends on two ideas: first, that in any convex polyhedron, the total angular deficit (the amount by which the sum of angles at each vertex fall short of $2\pi$, summed over all vertices) is exactly $4\pi$; this is Descartes' theorem on total angular defect, a discrete form of the Gauss--Bonnet theorem. Second, in a monohedral isosceles convex polyhedron, there are (usually) only a small number of distinct types of vertex, as follows:
\begin{itemize}
\item If a vertex has only the apexes of isosceles triangles incident to it, we call it a \emph{pyramidal vertex}.
\item If a vertex has two bases of isosceles triangles and one or more apexes incident to it, we call it a \emph{semipyramidal vertex}.
\item If a vertex has only four bases of isosceles triangles incident to it, we call it a \emph{basic vertex}.
\item If a vertex has four bases and one apex incident to it, we call it a \emph{semibasic vertex}.
\end{itemize}

As a shorthand, we call a convex polyhedron with congruent isosceles faces whose vertices are only of these four types \emph{well-behaved}.

\begin{observation}
At any vertex of a monohedral isosceles convex polyhedron, the number of incident base angles of faces is even, and these angles can be grouped together in pairs, adjacent across a shared base edge of two isosceles triangles.
\end{observation}

\begin{observation}
\label{obs:few-bases}
In a monohedral isosceles convex polyhedron, a vertex incident to four or more base angles can be incident to at most one apex angle, so the well-behaved vertex types exhaust the possible vertices that can be incident to zero, two, or four base angles.
\end{observation}

\begin{proof}
A vertex of a monohedral isosceles convex polyhedron with four or more base angles can have only one incident apex angle, because four base angles and two apex angles together add to $2\pi$, too much for the vertex of a convex polyhedron. The well-behaved vertex types allow arbitrarily many apex angles together with zero or two base angles, and zero or one apex angles with four base angles, so they include all the vertices that have at most four base angles and obey this constraint.
\end{proof}

\begin{lemma}
\label{lem:nwb-dihedral}
In a semipyramidal vertex incident to exactly two apex angles of isosceles triangles, with all incident triangles congruent and having apex angle greater than $\pi/2$, the dihedral angle between the two incident apex angles is less than $\pi/2$.
\end{lemma}

\begin{proof}
For two isosceles triangles with apex angle $\theta>\pi/2$ that meet on a shared side and lie on perpendicular planes, the angle between their other two sides can be calculated as $\cos^{-1}(\cos^2(\pi-\theta))>\pi-\theta$. However, two base angles of the same triangle can together span an angle of at most $\pi-\theta$, not enough to connect two apex angles with this dihedral. Therefore, in order to form a connected surface at the vertex, the dihedral must be smaller.
\end{proof}

Monohedral isosceles convex polyhedra with more than four base angles at some vertices do exist. For instance, the Kleetopes of the regular tetrahedron, octahedron, or icosahedron have six, eight, or ten bases meeting at some vertices respectively. However, there are only finitely many such cases:

\begin{theorem}
\label{thm:nwb-sharp}
There are only a finite number of triangulations that can be realized as non-well-behaved monohedral isosceles convex polyhedra.
\end{theorem}

\begin{proof}
Let $P$ be such a polyhedron, with $n$ faces.  By \autoref{obs:few-bases}, $P$ must have a vertex $v$ at which six or more base angles meet. Therefore, the base angle must be less than $\pi/3$, from which it follows that the apex angle must be greater than $\pi/3$.

In $P$, as in any icosahedral polyhedron) the number of base angles of faces is exactly twice the number of apex angles of faces, and because the large apex angle allows at most five apex angles to meet at a vertex, there must be $\Omega(n)$ faces with at least half as many apex angles as base angles. According to \autoref{obs:few-bases}, such a vertex can only be pyramidal or semi-pyramidal, with at most five incident faces. In order to have total angular defect $4\pi$, these faces need to have angular defect $O(1/n)$. We distinguish the following cases:
\begin{itemize}
\item A semipyramidal vertex with one incident apex angle has angular defect exactly $\pi$, which can only be $O(1/n)$ if $n$ is itself bounded.
\item If there are $\Omega(n)$ semipyramidal vertices with two incident apex angles, then to achieve angular defect $O(1/n)$ for each of these vertices, the apex angle must be $\pi-2\delta$ and the base angle must be $\delta$ for some $\delta=O(1/n)$.

Because these vertices use equal numbers of apex and base angles, but the faces of the polyhedron have twice as many base angles as apex angles, there must also be a linear number of base angles that participate in other vertices with zero or one incident apex angles. If there are $k$ of these other vertices, then each apex angle contributes less than $\pi$ to the total angle of each of these other vertices, and all of the base angles together contribute $O(1)$ to the total angle, so the total angle is $\le k\pi+O(1)$, but in order to achieve total angular deficit at most $4\pi$ the total angle must be $\ge (2k-4)\pi$. Combining these two inequalities shows that $k=O(1)$ and that, therefore, at least one of these other vertices has $\Omega(n)$ incident base angles.

At a vertex with $\Omega(n)$ incident base angles, each pair of consecutive base angles that meet along a side of their faces (rather than sharing a base edge) must have a semipyramidal vertex with two incident apex angles at the other end of that side. When $n$ is large enough that the apex angle is greater than $\pi/2$, by \autoref{lem:nwb-dihedral}, the dihedral angle between these two consecutive base angles must be less than $\pi/2$. But a vertex of a convex polyhedron can have at most three dihedral angles less than $\pi/2$, so $n$ must be bounded.

\item If there are $\Omega(n)$ pyramidal vertices with three incident apex angles, then to achieve angular defect $O(1/n)$ for each of these vertices, the apex angle must be $2\pi/3-2\delta$ and the base angle must be $\pi/6+\delta$ for some $\delta=O(1/n)$. However, the base angles of the faces must also appear in a linear number of non-pyramidal vertices, and for these apex and base angles, all other vertex types have angular deficit at least $\pi/6-O(\delta)$, bounded away from zero. Therefore, to keep the total angular defect below $4\pi$, $n$ must be bounded.
\item If there are $\Omega(n)$ pyramidal vertices with four incident apex angles, or semipyramidal vertices with three incident apex angles, then to achieve angular defect $O(1/n)$ for each of these vertices, the apex angle must be $\pi/2-2\delta$ and the base angle must be $\pi/4+\delta$ for some $\delta=O(1/n)$. Because vertices of these types use fewer than twice as many base angles as apex angles, the base angles of the faces must also appear in a linear number of vertices of other types, and for these apex and base angles, all other vertex types have angular deficit at least $\pi/4-O(\delta)$, bounded away from zero. Therefore, to keep the total angular defect below $4\pi$, $n$ must be bounded.
\item If there are $\Omega(n)$ pyramidal vertices with five incident apex angles, then to achieve angular defect $O(1/n)$ for each of these vertices, the apex angle must be $2\pi/5-2\delta$ and the base angle must be $3\pi/10+\delta$ for some $\delta=O(1/n)$. Because vertices of these types use fewer than twice as many base angles as apex angles, the base angles of the faces must also appear in a linear number of vertices of other types, and for these apex and base angles, all other vertex types have angular deficit at least $\pi/5-O(\delta)$, bounded away from zero. Therefore, to keep the total angular defect below $4\pi$, $n$ must be bounded.
\item There can be no pyramidal vertices with six or more incident apex angles, or semipyramidal vertices with four or more incident apex angles, because the total angle of such a vertex would be greater than $2\pi$.
\end{itemize}
No case allows $n$ to be arbitrarily large, so it must be bounded, and the number of triangulations with $n$ faces must also be bounded.
\end{proof}

\begin{figure}[t]
\centering\includegraphics[width=0.65\textwidth]{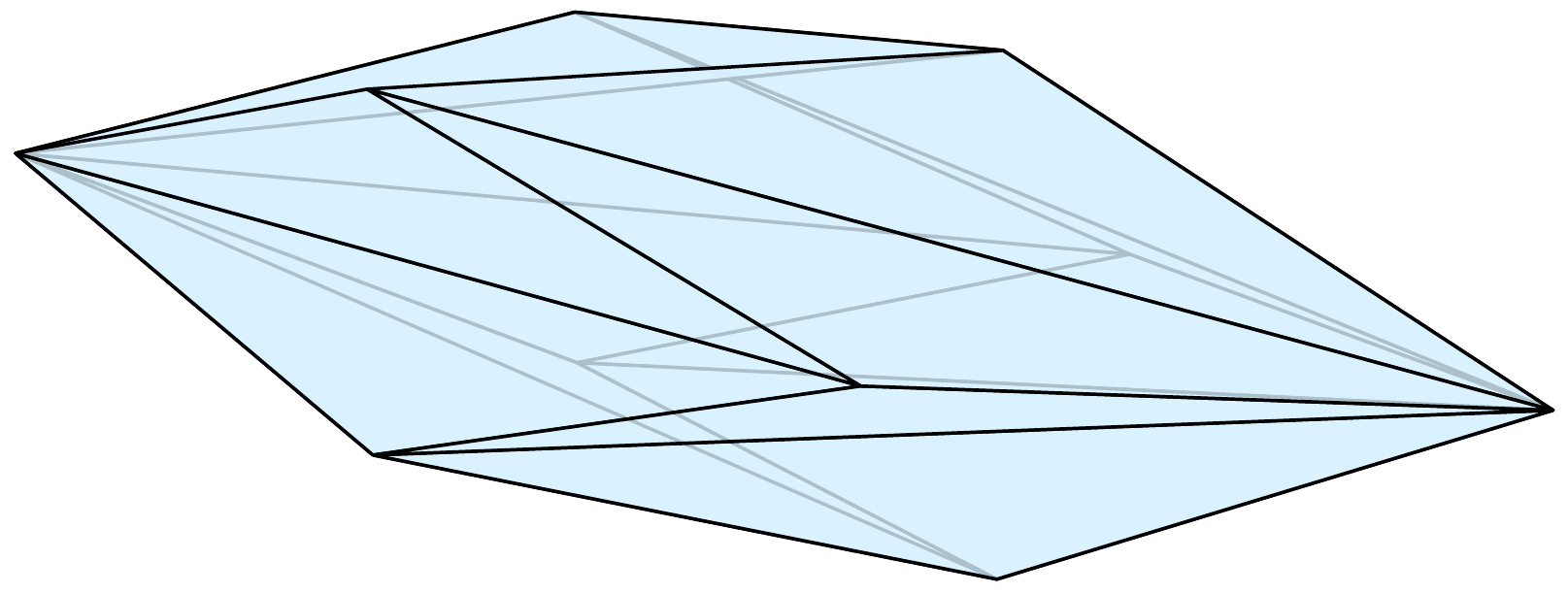}
\caption{A polyhedron with 16 faces, part of an infinite family of non-well-behaved non-convex monohedral isosceles bipyramids in which every vertex has positive angular deficit.}
\label{fig:nwb}
\end{figure}

We remark that this result applies only to convex realizations. There exist non-convex non-well-behaved monohedral isosceles polyhedra, combinatorially equivalent to bipyramids, with arbitrarily many faces. It is even possible to realize these polyhedra in such a way that all vertices have positive angular defect (\autoref{fig:nwb}).

\begin{theorem}
\label{thm:wb-sharp}
A monohedral isosceles convex polyhedron with $n$ faces must have apex angle $O(1/n)$.
\end{theorem}

\begin{proof}
By \autoref{thm:nwb-sharp} we may assume without loss of generality that the polyhedron is well-behaved, because the finitely many non-well-behaved possibilities do not affect the asymptotic statement of the theorem.

Because the polyhedron has twice as many base angles as apex angles, and (by the assumption that it is well-behaved) each vertex has $O(1)$ base angles, it follows that there are $\Omega(n)$ vertices that have at least twice as many incident base angles as apex angles. To keep the total angular deficit equal to $4\pi$, $\Omega(n)$ of these vertices must have deficit $O(1/n)$. This rules out the semipyramidal vertices with one incident apex angle, for sufficiently large $n$, because they have deficit exactly $\pi$. The only remaining possibilities are the basic and semibasic vertices.  If the base angle is $\theta$ and the apex angle is $\pi-2\theta$, then a basic vertex has angular deficit $2\pi-4\theta$ and a semibasic vertex has angular deficit $\pi-2\theta$. This can only be $O(1/n)$ if $\theta=\pi/2-O(1/n)$, for which the apex angle is $O(1/n)$.
\end{proof}

\subsection{Graph-theoretic properties}

The infinite families of monohedral isosceles convex polyhedra listed at the start of this section
all have many basic or semi-basic vertices, but few pyramidal or semipyramidal vertices. This can be formalized and proven more generally: 

\begin{lemma}
\label{lem:few-pyramidal}
In a well-behaved monohedral isosceles convex polyhedron, there is a nonzero but bounded number of pyramidal or semipyramidal vertices.
\end{lemma}

\begin{proof}
The number of apex angles is half the number of base angles, but the basic and semi-basic vertices use more than twice as many base angles as apex angles, so the leftover apex angles must be incident to pyramidal or semipyramidal vertices. If there are $k$ pyramidal or semipyramidal vertices, then the base angles incident to these vertices contribute total angle less than $\pi k$ (less than $\pi$ for each pair of base angles incident to a semipyramidal vertex) and the apex angles contribute $O(1)$ (because there are $\le n$ of them, each with angle $O(1/n)$. Therefore, the total angle is at most $\pi k + O(1)$. However, the total angle must be at least $2\pi k-4\pi$ in order to achieve total angular deficit $4\pi$ for the whole polyhedron. Combining these inequalities, $k$ must be $O(1)$.
\end{proof}

We can use this fact to show that the graphs of monohedral isosceles convex polyhedra have small dominating sets. Recall that a dominating set, in a graph, is a subset of the vertices such that every vertex in the graph either belongs to this set or has a neighbor in the set.

\begin{theorem}
If $G$ is a triangulation that can be realized as a monohedral isosceles convex polyhedron, then $G$ has a dominating set of size $O(1)$.
\end{theorem}

\begin{proof}
If the realization is not well-behaved, then $G$ has $O(1)$ vertices, and we may take the entire vertex set as the dominating set. Otherwise, we claim that the pyramidal and semipyramidal vertices form a dominating set, which by \autoref{lem:few-pyramidal} has size $O(1)$. By the assumption that the realization is well-behaved, every vertex is pyramidal or semipyramidal (and belongs to this set), or is basic or semibasic. In a basic or semibasic vertex, there are four incident base angles, and at most one apex angle to separate them, so some two base angles must be adjacent across a shared side of their two faces rather than across a shared base edge. The vertex adjacent to the basic or semibasic vertex at the other end of this shared side has two incident apex angles, so it must be pyramidal or semipyramidal.
\end{proof}

\begin{theorem}
If $G$ is a triangulation that can be realized as a monohedral isosceles convex polyhedron, then $G$ has unweighted graph diameter $O(1)$.
\end{theorem}

\begin{proof}
Choose a starting vertex $v$ within the small dominating set for $G$, and partition the vertices into subsets $S_i$ of the vertices at distance $i$ from $v$ (the layers of $G$ in a breadth-first search from $v$). Then at most two consecutive layers can be disjoint from the dominating set, because three or more consecutive layers would leave the vertices in the middle layer without a neighbor in the dominating set. Therefore, the number of layers is at most three times the size of the dominating set.
\end{proof}

\section{Conclusions and open problems}

We have shown that not every triangulated polyhedral graph has a convex isosceles realization, but that a wide family of polyhedral graphs (including all the ones that we use to prove have no convex isosceles realizations) do have non-convex isosceles realizations. It would be of interest to find triangulated polyhedral graphs for which even a non-convex but non-self-crossing isosceles realization is impossible. More ambitiously, we would like to understand the computational complexity of testing whether a (convex or non-convex) isosceles realization exists, and of finding one when it exists, but such a result seems far beyond the methods we have used in this paper.

We have also added to the zoo of known infinite families of polyhedra with congruent isosceles faces. As well as these infinite families, there also some number of sporadic examples of such polyhedra, presumably many more than the eight convex deltahedra (convex polyhedra with congruent equilateral triangle faces). Is it possible to classify all such polyhedra rather than merely finding more examples? Can our graph-theoretic analysis assist in this classification?

The example of Gr\"unbaum's polytope (\autoref{fig:grunbaum}, right) shows that there are additional families of convex polyhedra with edges of two lengths (necessarily, with all faces isosceles or equilateral) than the families of monohedral isosceles polyhedra already listed. Can these families be classified? To what extent can our structural analysis of monohedral isosceles polyhedra be extended to them? Our triangulations that cannot be realized as isosceles convex polyhedra show that, for convex realizations, at least three edge lengths may sometimes be needed. Is the minimum number of edge lengths of convex realizations of polyhedral graphs bounded or unbounded? What about for non-convex realizations?

\section*{Acknowledgements}
We thank Joseph Malkevitch for directing our attention to his \emph{Consortium} paper, and Philip Rideout for the svg3d Python package with which many of the figures were made.

\bibliographystyle{plainurl}
\bibliography{isosceles}
\end{document}